\documentclass[a4paper, american, USenglish]{lipics-v2021}

\usepackage{doi}
\usepackage{algorithm}
\usepackage[noend]{algpseudocode}
\usepackage{tikz}
\usepackage{adjustbox}
\usetikzlibrary{patterns}
\usetikzlibrary{decorations.pathreplacing,calligraphy}
\usepackage{wref}

\newcommand{\Bentropy}{\mathcal{B}}
\newcommand{\Bio}{\mathcal{B}_{\mathrm{I/O}}}
\newcommand{\ESymbol}{\mathbb{E}}
\newcommand*\ie{\mbox{i.\hspace{.2ex}e.}}
\newcommand*\eg{\mbox{e.\hspace{.2ex}g.}}

\newcommand\algname[1]{#1\xspace}

\let\epsilon\varepsilon

\hideLIPIcs
\nolinenumbers

\title{Deterministic Cache-Oblivious Funnelselect}

\author{Gerth Stølting Brodal}{Aarhus University, Denmark}{gerth@cs.au.dk}{https://orcid.org/0000-0001-9054-915X}{Independent Research Fund Denmark, grant~9131-00113B.}

\author{Sebastian Wild}{University of Liverpool, UK}{wild@liverpool.ac.uk}{https://orcid.org/0000-0002-6061-9177}{Engineering and Physical Sciences Research Council grant EP/X039447/1.}

\authorrunning{G.\,S. Brodal and S. Wild}

\Copyright{Gerth Stølting Brodal and Sebastian Wild}

\ccsdesc[500]{Theory of computation~Design and analysis of algorithms}

\keywords{Multiple selection, cache-oblivious algorithm, entropy bounds}

\ArticleNo{1}

\begin{document}

\maketitle


\begin{abstract}
    In the multiple-selection problem one is given an unsorted array~$S$ of $N$ elements and an array of $q$ query ranks $r_1<\cdots<r_q$, and the task is to return, in sorted order, the $q$ elements in $S$ of rank $r_1, \ldots, r_q$, respectively. 
    The asymptotic deterministic comparison complexity of the problem was settled by Dobkin and Munro [JACM~1981]. 
    In the I/O model an optimal I/O complexity was achieved by Hu \textit{et al}.~[SPAA~2014]. 
    Recently [ESA~2023], we presented a \emph{cache-oblivious} algorithm with matching I/O complexity, named \emph{funnelselect}, since it heavily borrows ideas from the cache-oblivious sorting algorithm \emph{funnelsort} from the seminal paper by Frigo,  Leiserson, Prokop and Ramachandran [FOCS~1999].
    \algname{Funnelselect} is inherently randomized as it relies on sampling for cheaply finding many good pivots.
    In this paper we present \emph{deterministic funnelselect}, achieving the same optional I/O complexity cache-obliviously without randomization. 
    Our new algorithm essentially replaces a single (in expectation) reversed-funnel computation using random pivots by a recursive algorithm using multiple reversed-funnel computations. 
    To meet the I/O bound, this requires a carefully chosen subproblem size based on the entropy of the sequence of query ranks; 
    \algname{deterministic funnelselect} thus raises distinct technical challenges not met by randomized \algname{funnelselect}.
    The resulting worst-case I/O bound is $O\bigl(\sum_{i=1}^{q+1} \frac{\Delta_i}{B} \cdot \log_{M/B} \frac{N}{\Delta_i} + \frac{N}{B}\bigr)$, where $B$ is the external memory block size, $M\geq B^{1+\epsilon}$ is the internal memory size, for some constant $\epsilon>0$, and $\Delta_i = r_{i} - r_{i-1}$ (assuming $r_0=0$ and $r_{q+1}=N + 1$). 
\end{abstract}


\section{Introduction}

We present the first optimal deterministic cache-oblivious algorithm for the multiple-selection problem. 
In the multiple-selection problem one is given an unsorted array~$S$ of $N$ elements and an array $R$ of $q$ query ranks in increasing order $r_1<\cdots<r_q$, and the task is to return, in sorted order, the $q$ elements of $S$ of rank $r_1, \ldots, r_q$, respectively;  (see \wref{fig:example} for an example).

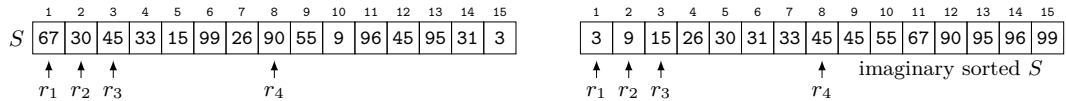
\begin{figure}[tbp]
	\adjustbox{max width=\linewidth}{\begin{tikzpicture}[scale=.45]
	\small
        \node[anchor=east] at (0.5,0) {$S$};
		\foreach [count=\i] \x in {67, 30, 45, 33, 15, 99, 26, 90, 55, 9, 96, 45, 95, 31, 3} {
			\node[scale=.9] at (\i,0) {\ttfamily \x} ;
			\draw (\i-.5,-.5) rectangle ++(1,1) ;
			\node[scale=.6] at (\i,.8) {\ttfamily \i} ;
		}
		\foreach [count=\i] \r in {1, 2, 3, 8} {
			\draw[latex-] (\r,-.7) -- ++(0,-1) node[inner sep=3pt,fill=white] {$r_{\i}$} ;
		}
		\begin{scope}[shift={(17,0)}]
			\foreach [count=\i] \x in {3, 9, 15, 26, 30, 31, 33, 45, 45, 55, 67, 90, 95, 96, 99} {
				{\node[scale=.9] at (\i,0) {\ttfamily \x} ;}
				\draw (\i-.5,-.5) rectangle ++(1,1) ;
				\node[scale=.6] at (\i,.8) {\ttfamily \i} ;
			}
			\foreach [count=\i] \r in {1, 2, 3, 8} {
				\draw[latex-] (\r,-.7) -- ++(0,-1) node[inner sep=3pt,fill=white] {$r_{\i}$} ;
			}
			\node at (12,-1) {\footnotesize imaginary sorted $S$} ;
		\end{scope}
	\end{tikzpicture}}
	\caption{
		Example input with $N=15$, $q=4$ and $R[1..q]=[1,2,3,8]$. 
		The expected output \texttt 3, \texttt9, \texttt{15}, \texttt{45} is obvious from the sorted array (right).
		(The sorted array is for illustration only; the goal of efficient multiple-selection algorithms is to avoid ever fully sorting the input.)
	}
	\label{fig:example}
\end{figure}

On top of immediate applications, the multiple-selection problem is of interest as it gives a natural common generalization of (single) selection by rank (using a single query rank $r_1=r$) and fully sorting an array (corresponding to selecting every index as a query rank, \ie, $q=N$ and $r_i=i$ for $i=1,\ldots,N$).
It thus allows us to quantitatively study the transition between these two foundational problems, which are of different complexity and each have their distinct set of algorithms.
For example, the behavior of selection and sorting with respect to external memory is quite different: For single selection, the textbook \algname{median-of-medians algorithm}~\cite{BFPRT73} simultaneously works with optimal cost in internal memory, external memory, and the cache-oblivious model (models are defined below).  For sorting, by contrast, the introduction of each model required a substantially modified algorithm to achieve optimal costs: 
Standard \algname{binary mergesort} is optimal in internal memory, but requires $\approx M/B$-way merging to be optimal in external memory, where $M$ is the internal memory size and $B$ the external memory block size, measured in elements~\cite{AggarwalVitter88}; achieving the same cache obliviously, \ie, without knowledge of $B$ and $M$, requires the judiciously chosen buffer sizes from the recursive constructions of \algname{funnelsort}~\cite{FrigoLPR99}.

Since multiple selection simultaneously generalizes both problems, it is not surprising that also here subsequent refinements were necessary going from internal to external to cache-oblivious; the most recent result being our algorithm \algname{funnelselect}~\cite{BrodalWild23}.
However, all algorithms mentioned above for single selection and sorting are \emph{deterministic}. By constrast, \algname{funnelselect} is inherently relying on randomization and known deterministic external-memory algorithms~\cite{BarbayGuptaRaoSorenson2016,HuTaoYangZhou2014} are crucially relying on the knowledge of $M$ and $B$.
Prior to this work it thus remained open whether a single deterministic cache-oblivious algorithm exists that smoothly interpolates between selection and sorting without having to resort to randomization.

In this paper, we answer this question in the affirmative.
Our algorithm \emph{deterministic funnelselect} draws on techniques from cache-oblivious sorting (\algname{funnelsort}) and existing multiple-selection algorithms, but it follows a rather different approach to our earlier randomized algorithm~\cite{BrodalWild23} and previous (cache-conscious) external-memory algorithms. 
A detailed comparison is given below.
%

\subsection{Model of computation and previous work}

Our results are in the cache-oblivious model of Frigo, Leiserson, Prokop and Ramachandran~\cite{FrigoLPR12},
a hierarchical-memory model with an infinite external memory and an internal memory of capacity $M$ elements, 
where data is transferred between internal and external memory in blocks of $B$ consecutive elements. Algorithms are compared by their I/O cost, \ie, the number of block transfers or \emph{I/O}s (input/output operations). This is similar to the external-memory model by Aggarwal and Vitter~\cite{AggarwalVitter88}.
Crucially, in the cache-oblivious model, algorithms \emph{do not know} $M$ and $B$ and I/Os are assumed to be performed automatically by an optimal (offline) paging algorithm.
Cache-oblivious algorithms hence work for any parameters $M$ and $B$, and they even adapt to multi-level memory hierarchies (under certain conditions~\cite{FrigoLPR12}). 
	
\begin{table}[tbp]
  \caption{Algorithms for selection and multiple selection. CO = cache-oblivious, $\ESymbol$ = expected, wc = worst-case bounds. Note that Barbay \textit{et al.} assume a tall cache $M \ge B^{1+\epsilon}$, whereas Hu \textit{et al.} do not.}
  \centering
  \resizebox{\linewidth}!{
  \begin{tabular}{lcccl}
  	\hline
  	\textbf{Reference}\strut                                 &            &      \textbf{Comparisons}      &      \textbf{I/Os}     & \textbf{Comments}                       \\
    \hline
  	\textsl{\textbf{Single selection}}\rule{0pt}{12pt}       &            &                                &                        &                                         \\
  	Hoare~\cite{H61find}                                     & $\ESymbol$ &  $2\ln 2\Bentropy + 2 N+o(N)$  &         $O(N/B)$       & CO, randomized                          \\
  	Floyd \& Rivest~\cite{FloydRivest1975}                   & $\ESymbol$ &    $N+\min\{r,N{-}r\}+o(N)$    &         $O(N/B)$       & CO, randomized                          \\
  	Blum \textit{et al.}~\cite{BFPRT73}                      &     wc     &             $5.4305 N$         &         $O(N/B)$       & CO, deterministic                       \\
    Sch\"onhage \textit{et al.}~\cite{SchonhagePP76}         &     wc     &         $3 N + o(N)$           &            ?           & median, deterministic                   \\
    Dor \& Zwick~\cite{DorZwick1999}                         &     wc     &         $2.95 + o(N)$         &            ?           & median, deterministic                    \\[1ex]
  	\textsl{\textbf{Multiple selection}}                     &            &                                &                        &                                         \\
  	Chambers~\cite{Chambers71,Prodinger1995}                 & $\ESymbol$ &    $2\ln 2 \Bentropy+O(N)$     &  $O((\Bentropy+N)/B)$  & CO, randomized                          \\
  	Dobkin \& Munro~\cite{DobkinMunro81}                     &     wc     &       $3\Bentropy+O(N)$        & $O((\Bentropy+N) / B)$ & CO, deterministic                       \\
  	Kaligosi \textit{et al.}~\cite{KMMS05}                   &     wc     & $\Bentropy+o(\Bentropy)+O(N)$  & $O((\Bentropy+N) / B)$ & CO, deterministic                       \\
  	Hu \textit{et al.}~\cite{HuTaoYangZhou2014}              &     wc     &         $O(N \lg(q))$          & $O(N/B \log_{M/B}(q/B))$ & deterministic \\
                                                             &     wc     &       $O(\Bentropy + N)$       &     $O(\Bio + N/B)$    & (from closer analysis)                  \\
  	Barbay \textit{et al.}~\cite{BarbayGuptaRaoSorenson2016} &     wc     & $\Bentropy+o(\Bentropy)+O(N)$  &     $O(\Bio + N/B)$    & online, determ., $M \ge B^{1+\epsilon}$ \\
    Brodal \& Wild~\cite{BrodalWild23}                       & $\ESymbol$ &       $O(\Bentropy + N)$       &     $O(\Bio + N/B)$    & CO, randomized, $M \ge B^{1+\epsilon}$  \\    
    \textit{This paper}                                      &     wc     &       $O(\Bentropy + N)$       &     $O(\Bio + N/B)$    & CO, deterministic, $M \ge B^{1+\epsilon}$  \\    
    \hline
  \end{tabular}}	
  \label{tab:results}
\end{table}

The multiple-selection problem was first formally addressed by Chambers~\cite{Chambers71}, who considered it a generalization of \algname{quickselect}~\cite{H61find}. Prodinger~\cite{Prodinger1995} proved that Chambers' algorithm achieves an optimal \emph{expected} running time up to constant factors: $O(\Bentropy+N)$, where $\Bentropy = \sum_{i=1}^{q+1} \Delta_i \lg \frac{N}{\Delta_i}$ with $\Delta_i = r_{i} - r_{i-1}$, for $1 \leq i \leq q + 1$, 
assuming $r_0=0$ and $r_{q+1}=N + 1$, and $\lg$ denoting the binary logarithm. We call $\Bentropy$ the \emph{(query-rank) entropy} of the sequence of query ranks~\cite{BarbayGuptaRaoSorenson2016}.
It should be noted that $\Bentropy+N=O(N(1+\lg q))$, but the latter bound does not take the location of query ranks into account; for example, if $q=\Theta\bigl(\sqrt{n}\bigr)$ queries are in a range of size $O(N/\lg N)$, \ie, $r_q - r_1 = O(N/\lg N)$, then the entropy bound is $O(N)$ whereas the latter $N(1+\lg q) = \Theta(N\lg N)$.

Dobkin and Munro~\cite{DobkinMunro81} showed that $\Bentropy-O(N)$ comparisons are necessary to find all ranks $r_1,\ldots,r_q$ (in the worst case).
Deterministic algorithms with that same $O(\Bentropy+N)$ running time are also known~\cite{DobkinMunro81,KMMS05},
but as for single selection, the deterministic algorithms were presented later than the randomized algorithms and require more sophistication.
Multiple selection in external-memory was studied by Hu \textit{et al}.~\cite{HuTaoYangZhou2014} and Barbay \textit{et al}.~\cite{BarbayGuptaRaoSorenson2016}. Their algorithms have an I/O cost of $O\bigl(\Bio + \frac{N}{B}\bigr)$, where the \emph{``I/O entropy''} $\Bio = \frac{\Bentropy}{B \lg(M/B)}$. 
An I/O cost of $\Omega(\Bio) - O(\frac NB)$ is known to be necessary~\cite{BarbayGuptaRaoSorenson2016,BrodalWild23}.
A more comprehensive history of the multiple-selection problem appears in~\cite{BrodalWild23}; \wref{tab:results} gives an overview.

We note that many existing time- and comparison-optimal multiple-selection algorithms are actually already cache oblivious, but they are not optimal with respect to the number of I/Os performed when analyzed in the cache-oblivious model (the obtained I/O bounds are a factor $\lg(M/B)$ away from being optimal).

\subsection{Result}

Our main result is the cache-oblivious algorithm \emph{deterministic funnel\-select} achieving the following efficiency (see \wref{thm:deterministic-funnelselect} for the full statement and proof).
	
\begin{theorem}
\label{thm:upper-bound}
	There exists a deterministic cache-oblivious algorithm solving the 
	multiple-selection problem using 
    $O(\Bentropy + N)$ comparisons and
    $O\bigl(\Bio + \frac{N}{B}\bigr)$ I/Os in the worst case,
    assuming a tall cache $M\geq B^{1+\epsilon}$.
\end{theorem}
	
At the high level, our algorithm uses the standard overall idea of a recursive partitioning algorithm and pruning recursive calls containing no rank queries, an idea dating back to the first algorithm by Chambers~\cite{Chambers71}. 
In the cache-aware external-memory model, I/O efficient algorithms are essentially obtained by replacing binary partitioning (as used in~\cite{Chambers71}) by an external-memory $\Theta(M/B)$-way partitioning \cite{BarbayGuptaRaoSorenson2016,HuTaoYangZhou2014}.
Unfortunately, in the cache-oblivious model this is not possible, since the parameters $M$ and $B$ are unknown to the algorithm.
To be I/O efficient in the cache oblivious model, both our previous algorithm \algname{randomized funnelselect}~\cite{BrodalWild23} and our new algorithm \algname{deterministic funnelselect} apply a cache-oblivious multi-way $k$-partitioner to distribute elements into $k$ buckets given a set of $k-1$ pivot elements, essentially reversing the computation done by the $k$-merger used by \algname{funnelsort}~\cite{FrigoLPR99}. The $k$-partitioner is a balanced binary tree of $k-1$ pipelined binary partitioners.

The key difference between our randomized and deterministic algorithms is that in our randomized algorithm we use a single $N^{\Theta(\epsilon)}$-way partitioner using randomly selected pivots and
truncate work inside the partitioner for subproblems that (with high probability) will not contain any rank queries.  This is done by estimating the ranks of the pivots through sampling and pruning subproblems estimated to be sufficiently far from any query ranks. 
In our deterministic version, we choose $k$ smaller and deterministically compute pivots, such that all elements are pushed all the way down through a $k$-partitioner without truncation (eliminating the need to know the (approximate) ranks of the pivots before the $k$-partitioning is finished), while we choose $k$ such that the buckets with unresolved rank queries (that we  have to recursive on) in total contain \emph{at most half} of the elements.
To compute~$k$, we apply a linear-time \emph{weighted}-median finding algorithm on $\Delta_1,\ldots,\Delta_{q+1}$.
While \algname{randomized funnelselect} can handle buckets with unresolved rank queries directly using sorting, \algname{deterministic funnelselect} needs to recursively perform multiple-selection on the buckets to achieve the desired I/O performance.

%
%
%
%
%
%
%

\section{Preliminaries}

Throughout the paper we assume that the input to a multiple-selection algorithm is given as two arrays $S[1..N]$ and $R[1..q]$, where $S$ is an unsorted array of $N$ elements from a totally ordered universe, and $R$ is a sorted array $r_1, \ldots, r_q$ of $q$ distinct query ranks, where $1 \leq r_1 < \cdots < r_q \leq N$. 
The array~$S$ is allowed to contain duplicate elements.
Our task is to produce/report an array of the $q$ order statistics $S_{(r_1)}, \ldots, S_{(r_q)}$, where $S_{(r)}$ is the $r$th smallest element in $S$, \ie, the element at index~$r$ in an array storing $S$ after sorting it.

Our new deterministic cache-oblivious multiple-selection algorithm makes use of the following three existing cache-oblivious results for single selection, weighted selection, sorting, and multi-way partitioning.

\begin{lemma}[Blum, Floyd, Pratt, Rivest, Tarjan~{\cite[Theorem~1]{BFPRT73}}]
\label{lem:selection}
    Selecting the $k$-th smallest element in an unsorted array of $N$ elements can be done with $O(N)$ comparisons and $O\bigl(1+\frac{N}{B}\bigr)$ I/Os in the cache-oblivious model.
\end{lemma}

\begin{remark}[Median of medians: I/O cost]
Although the original paper by Blum \textit{et al}.~\cite{BFPRT73} predates the cache-oblivious model~\cite{FrigoLPR99} by decades, analyzing the algorithm in the cache-oblivious model with a stack-oriented memory allocator gives a linear I/O cost, since the algorithm is based on repeatedly scanning geometrically decreasing subproblems. 
\end{remark}

\begin{remark}[Median of medians: duplicates]
The original algorithm in~\cite{BFPRT73} assumes that all elements are distinct, but the algorithm can be extended to handle duplicates (by performing a three-way partition of the elements into those less-than, equal-to, and greater-than a pivot, respectively), and to return a triple $S_{\leq}, p, S_{\geq}$, that is a partition of~$S$, where $p$ is the element of rank~$k$, $S_{\leq}$ are the elements of rank $1,\ldots,k-1$ in arbitrary order, and $S_{\geq}$ are the elements of rank $k+1,\ldots,|S|$ in arbitrary order (where duplicate elements are assigned consecutive ranks in an arbitrary order).
\end{remark}

In the \emph{weighted selection} problem we are giving an array of $N$ elements, each with an associated non-negative weight, and a target weight $W$, where the goal is to return the $k$-th smallest element, for the smallest possible $k$, where the sum of the weights of the $k$ smallest elements at least $W$.
A linear-time weighted-selection algorithm can be derived from the unweighted algorithm by Blum \textit{et al}.~\cite{BFPRT73}  (\wref{lem:selection})~-- as hinted by Shamos in~\cite{Shamos76} and spelled out in detail by Bleich and Overton~\cite{BleichOverton83}~-- by computing the weighted rank of the pivot. The weighted selection algorithm follows essentially the same recursion as~\cite{BFPRT73}, and it similarly follows that it is cache oblivious and performs $O\bigl(1+\frac{N}{B}\bigr)$ I/Os. 

\begin{lemma}[Bleich, Overton~\cite{BleichOverton83}]
\label{lem:weighted-selection}
    \algname{Weighted selection} in an unsorted array of $N$ weighted elements can be done with $O(N)$ comparisons and $O\bigl(1+\frac{N}{B}\bigr)$ I/Os in the cache-oblivious model.
\end{lemma}

\begin{lemma}[Frigo, Leiserson, Prokop, Ramanchandran~{\cite[Theorem~7]{FrigoLPR12}}, Brodal, Fagerberg~{\cite[Theorem~2]{BF02}}]
\label{lem:funnelsort}
    \algname{Funnelsort} sorts an array of $N$~elements using
    $O\bigl(\frac{N}{B}(1 + \log_M N)\bigr)$ I/Os in a cache-oblivious model with a tall-cache assumption $M \geq B^{1+\epsilon}$, for constant $\epsilon > 0$.
\end{lemma}

\begin{remark}[Tall and taller]
The original description of funnelsort by Frigo \textit{et al}.~\cite{FrigoLPR99} assumed the tall cache assumption $M=\Omega(B^2)$, whereas~\cite{BF02} observed that this could be relaxed to the weaker tall cache assumption $M=\Omega\bigl(B^{1+\epsilon}\bigr)$.
I/O optimality of funnelsort follows from a matching external-memory lower bound by Aggarwal and Vitter~\cite[Theorem~3.1]{AggarwalVitter88}.
\end{remark}

The key innovation in our previous randomized algorithm \algname{funnelselect}~\cite{BrodalWild23} is the \emph{$k$-partitioner} (\wref{fig:k-partitioner}), a cache-oblivious and I/O-efficient multi-way partitioning algorithm to distribute a batch of elements around $k-1$ given pivots into $k$ buckets; the precise characteristics are summarized in the following lemma.

\begin{figure}[tbh]
	\centering
		\newcommand{\PARTITIONER}[3]{
			\node[draw,%
					circle,minimum size=5.5mm,inner sep=0pt,
					fill=red!50!orange!80!black!20] (N#1) at (#2) {\scalebox{.7}{$P_{#1}$}};
			\path[draw,latex-] (N#1) ++(-#3, -1) -- (N#1) ;
			\path[draw,latex-] (N#1) ++(#3, -1) -- (N#1);
			\path[draw,latex-] (N#1) -- ++(0, 1);		
		}	
		\newcommand{\BUFFER}[4]{
			\fill[white] (#1) rectangle (0.5+#1+#2);
			\fill[blue!50!black!40] (#1+#3) rectangle (0.5+#1+#4);
			\draw[draw=black] (#1) rectangle (0.5+#1+#2);
		}
		\colorlet{ckcolor}{green!45!white}
		\colorlet{cgcolor}{teal!50!blue!20!white}
		\colorlet{darkred}{red!40!black}
		\colorlet{blackred}{black!50!darkred}
		\colorlet{lightred}{red!80!green!35!white}
		\colorlet{mediumred}{red!40!gray}
		\begin{tikzpicture}[xscale=0.65,yscale=.6,thick]
		\large
			{
			\begin{scope}[ckcolor!50!black,very thick,fill=ckcolor!50!black!10,dashed,rounded corners=10pt]%
			\foreach \x in {1,5,9,13} {
				\begin{scope}[shift={(\x,0)}]
					\filldraw
						(-2,-.65) -- (2,-.65) -- (1.5,2) -- (0,3.35) 
						-- (-1.5,2) -- cycle
					;
				\end{scope}
			}
			\begin{scope}[shift={(7,5.5)}]
				\filldraw
					(-6,-.65) -- (6,-.65) -- (4.5,2) -- (0,3.35) 
					-- (-4.5,2) -- cycle
				;
			\end{scope}
			\end{scope}
			}
			\foreach \i in {1, 3,...,16}	{ \PARTITIONER{\i}{\i-1,0}{0.5} }
			\foreach \i in {2, 6, 10, 14} { \PARTITIONER{\i}{\i-1,2.5}{1.0} }
			\foreach \i in {4,12}				 { \PARTITIONER{\i}{\i-1,5.5}{2} }
			\foreach \i in {8}						{ \PARTITIONER{\i}{\i-1,8}{4} }
			\foreach \i/\t in {0/1,1/1.5,2/0.9,3/1.2,4/1.5,5/0.5,6/0.7,7/1.3,8/2,9/1.2,10/1.6,11/0.6,12/0.3,13/0.8,14/1.4,15/1.1}
			{ \BUFFER{\i-0.75,-3}{2}{0}{\t} }
			\foreach \i/\b/\t in {0/0/0.5, 1/0/0.1, 2/0/0.3, 3/0/0.2, 4/0/0.3, 5/0/0.15, 6/0/0.25, 7/0/0.35} 
			{ \BUFFER{2*\i-0.25,1}{0.5}{\b}{\t} }
			\foreach \i/\b/\t in {0/0/0.4, 1/0.6/1, 2/0/0.25, 3/0/0.75}
			{ \BUFFER{4*\i+0.75,3.5}{1}{\b}{\t} }
			\foreach \i/\b/\t in {0/0.25/0.5,1/0.0/0.25} 
			{ \BUFFER{8*\i+2.75,6.5}{0.5}{\b}{\t} }
			\BUFFER{6.75,9}{3}{1}{3}
			\node[left,font=\itshape,blue!50!black] at (11, 10.5) {input array};
			\node[left,align=center,font=\itshape,blue!50!black] at (17.5, -2) {output\\ buckets};
			{
					\node[left,overlay,blue!50!black] at (15.5, 4) {$k^{d/2}$};
					\node[left,overlay,blue!50!black] at (14.25, 6.75) {$k^{d/4}$};
					\node[left,overlay,blue!50!black] at (17., 1.25) {$k^{d/4}$};
				\node[black!50,ckcolor!50!black,font=\itshape,overlay] at (-1,7.0) {$\sqrt k$-partitioner};
			 	\node[black!50,ckcolor!50!black,font=\itshape,overlay] at (-3.1,1.5) {$\sqrt k$-partitioners};
				{\node[blue!50!black] at (-1.5, 4) {\textit{middle buffers}};}
			}
			\draw[->,blue!50!black!10,line width=1.5em]	(17.25,11.5) -- ++(0,-9);
		\end{tikzpicture}
	
	\caption{A $k$-partitioner for $k=16$ buckets. Content in the buffers is shaded;
		buffers are filled bottom-to-top; when full, they are flushed and then consumed from the bottom.
		The figure shows the situation where the input buffer for $P_6$ is being flushed down to its children (by partitioning elements around pivot $P_6$). The flush at $P_6$ was triggered during flushing $P_4$'s input buffer, which in turn has been called while flushing $P_8$ (the input). 
		\protect\\
		Buffer sizes for the three internal levels are shown next to the buffers.
		$k$-partitioners are defined recursively from a $\sqrt k$-partitioner at the top, a collection of $\sqrt k$ middle buffers, and $\sqrt k$ further $\sqrt k$-partitioners, each partitioning from one middle buffer 
		to $\sqrt k$ output buffers.
		(All sizes here ignore floors and ceilings; for the precise definition valid for all $k$, see~\cite{BrodalWild23}.)
		}
	\label{fig:k-partitioner}
\end{figure}
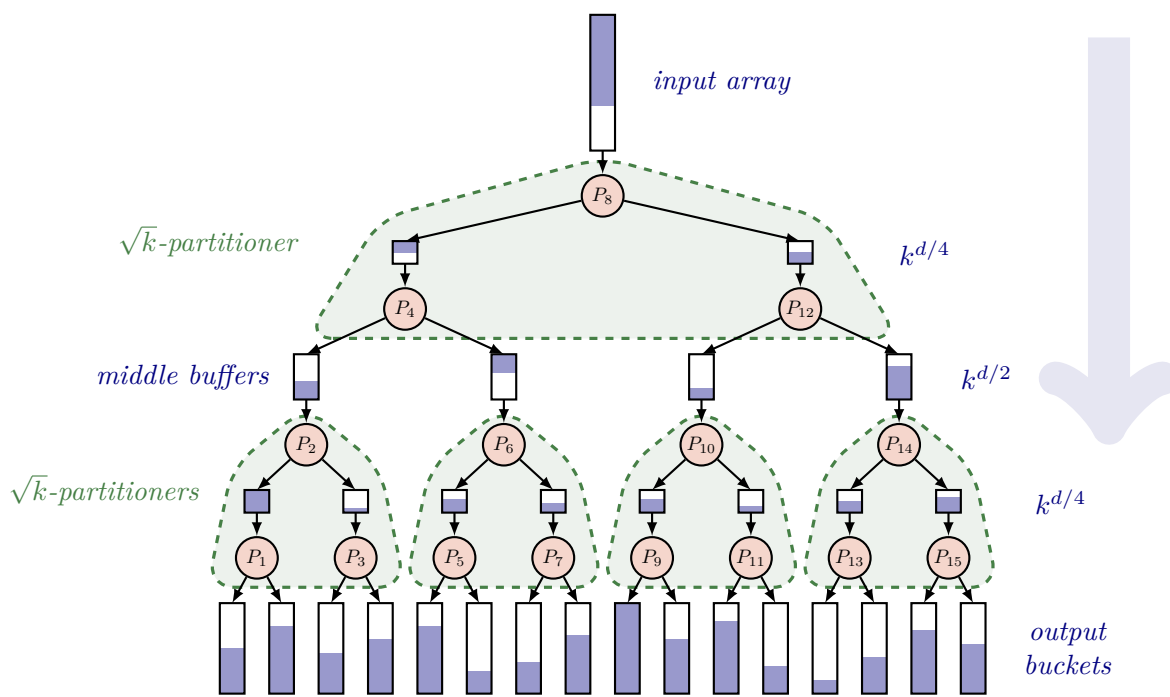

\begin{lemma}[Brodal and Wild~{\cite[Lemma~3]{BrodalWild23}}]
\label{lem:k-partitioner}
		Given an unsorted array of $N\geq k^d$ elements and $k-1$ pivots $P_1	\leq \cdots \leq P_{k-1}$, a $k$-partitioner can partition the elements into $k$ buckets $S_1,\ldots,S_k$, such all elements $x$ in bucket $S_i$ satisfy $P_{i-1} \leq x \leq P_i$. The algorithm is cache-oblivious and performs $O(N\lg k)$ comparisons and $O\bigl(k + \frac{N}{B}(1 + \log_M k)\bigr)$ I/Os, provided a tall-cache assumption $M \geq B^{1+\epsilon}$ and $d \geq \max\{ 1+2/\epsilon, 2\}$. 
    The working space for the $k$-partitioner (ignoring input and output buffers) is $O\bigl(k^{(d+1)/2}\bigr)$.
    This is also the time required to construct a $k$-partitioner (again ignoring input and output buffers).
\end{lemma}

The $k$-partitioners are structurally similar to the $k$-mergers from \algname{funnelsort} for merging $k$ runs cache obliviously. In~\cite{BrodalWild23} we pipeline the partitioning by essentially reversing the computations done by \algname{funnelsort}, and replace each binary merging node by a binary partitioning node.


\section{Deterministic multiple-selection}

In this section we present our deterministic cache-oblivious multiple-selection algorithm that performs optimal
$O(\Bentropy + N)$ comparisons and $O\bigl(\Bio + \frac{N}{B}\bigr)$ I/Os,
under a tall-cache assumption $M \geq B^{1 + \epsilon}$.
Detailed pseudo-code is given in \wref{alg:deterministic-multiple-selection} and \wref{alg:multi-partition}, and the basic idea is illustrated in \wref{fig:deterministic-multiple-selection}.

\begin{figure}[htb]
  \centering
  \begin{tikzpicture}[scale=0.195]
    \tikzstyle{every node}=[font=\small]
    \foreach \l/\r in { -4/11,  12/29, 30/47, 48/64 }    
        \draw[draw=black] (\l, 0) rectangle (\r, -1.5);
    \foreach \idx/\x in { 1/12, 2/30, 3/48 } {
        \draw[draw=black, fill=black] (\x-0.8, -0.2) rectangle (\x-0.2, -1.3);
        \node [anchor=north] at (\x-0.5, -1.5) {$P_{\idx}$};
    }
    \foreach \x in { 14, 27, 51, 62 }
        \draw[draw=black, fill=black] (\x-0.8, -0.3) rectangle (\x-0.2, -1.2);
    \foreach \idx/\x in { 1/4, 2/21, 3/38, 4/56 } {
        \node [anchor=north] at (\x-0.5, -1.5) {$S_{\idx}$};
    }
    \foreach \idx/\rank in { 0/-5, 1/13, 2/17, 3/22, 4/26, 5/47, 6/50, 7/55, 8/61, q+1/64 } {
       \draw [-latex, thick] (\rank+0.5,1.5) -- (\rank+0.5,0);
       \node[anchor=south, inner sep=1pt] at (\rank+0.5,1.8) {\smash{$r_{\idx}$}};
    }
    \foreach \x/\rank in { 13/r_2^{\min}, 26/r_2^{\max}, 50/r_4^{\min}, 61/r_4^{\max} } {
        \draw [-latex, thick] (\x+0.5,-3.2) -- (\x+0.5,-1.6);
        \node [anchor=north] at (\x+1.5, -4.2) {\smash{$\rank$}};
    }
    \foreach \l/\r in { 14/26, 51/61}
        \draw [draw=none, pattern=north east lines] (\l, -0.3) rectangle (\r, -1.2);
    \foreach \idx/\l/\r in { 1/-4/14, 2/14/18, 3/18/23, 4/23/27, 5/27/48, 6/48/51, 7/51/56, 8/56/62, 9/62/65 }
        \draw [decorate, decoration = {calligraphic brace,amplitude=5pt},line width=1.25pt] (\l+0.1,4) -- node [above=5pt, pos=0.5] {$\Delta_\idx$}  (\r-0.1, 4);
  \end{tikzpicture}
  \caption{Deterministic multiple selection. The partition of an array~$S$ into buckets $S_1,\ldots,S_4$ separated by pivots $P_1,\ldots,P_3$, and query ranks $r_1,\ldots,r_8$. In the example the maximum allowed bucket size is $\Delta=\Delta_1$, since 
  $\Delta_1+\Delta_2+\Delta_3+\Delta_4+\Delta_6+\Delta_7+\Delta_8+\Delta_9 \geq |S|/2+1$
  and $\Delta_2+\Delta_3+\Delta_4+\Delta_6+\Delta_7+\Delta_8+\Delta_9 < |S|/2+1$.
  Black squares are pivots and the shaded regions in buckets are the subproblems to recurse on.}
  \label{fig:deterministic-multiple-selection}
\end{figure}
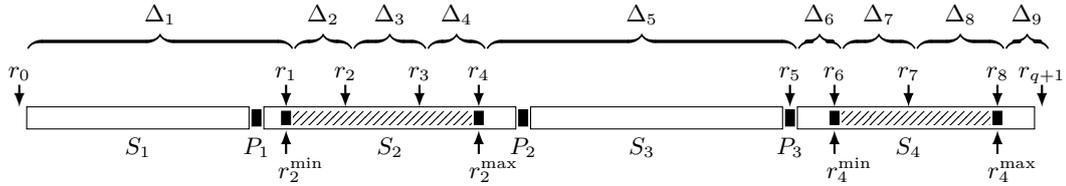

Given a tall-cache assumption $M \geq B^{1+\epsilon}$,
we let $d = \max\{1+2/\epsilon,2\}$. The algorithm follows the general idea of making a recursive multi-way partition of the array 
of elements and to only recurse on subproblems with unresolved rank queries. 
For two consecutive query ranks $r_{i-1}$ and $r_i$, we say that the $\Delta_i=r_i-r_{i-1}$ elements of rank $r_{i-1}+1,\ldots,r_i$ are in a \emph{gap} of size $\Delta_i$.
We choose a parameter $\Delta$, such that at least half of the elements are in gaps of size~$\leq \Delta$ and simultaneously at least half (rounded down) of the elements are in gaps of size~$\geq \Delta$. 
To compute $\Delta$ (\wref{alg:deterministic-multiple-selection}, line~4), we compute $\Delta_i = r_i-r_{i-1}$ by a scan over the query ranks $r_1,\ldots,r_q$ (and $r_0=0$ and $r_{q+1}=N+1$), and perform \algname{weighted selection} (\wref{lem:weighted-selection}) among $\Delta_1,\ldots,\Delta_{q+1}$, where $\Delta_i$ has weight~$w_i=\Delta_i$, and return the smallest $\Delta$ where $\sum_{\Delta_i \leq \Delta}  w_i \geq N/2+1$.

\begin{algorithm}[htp]
  \begin{algorithmic}[1]
    \Procedure{DeterministicFunnelselect}{$S[1..N]$, $R[1..q]$}
    \If{$q > 0$}
        \State $\Delta_i \gets R[i] - R[i-1]$ for $i=1,\ldots,q+1$, assuming $R[0]=0$ and $R[q+1]=N+1$
        \State $\Delta \gets \min \bigl\{
                \Delta_i \in \{\Delta_1,\ldots,\Delta_{q+1}\} 
                \bigm| \sum_{j \in \{1,\ldots,{q+1}\} : \Delta_j \leq \Delta_i} \Delta_j \geq N/2 + 1
            \bigr\}$
        \If{$(2N)^d \geq \Delta^{d + 1}$ \textbf{or} $N^{1+\frac{1}{1+\epsilon}} \geq \Delta^2$}
          \Comment{$\Bio = \Omega(\mathrm{Sort}_{M,B}(N))$} 
          \State $S$ $\gets$ \textsc{Funnelsort}$(S)$
          \State Report $S[R[1]], \ldots, S[R[q]]$
        \Else
        \State $(P_1,\ldots,P_{k-1}), (S_1,\ldots,S_k)$ $\gets$ \textsc{MultiPartition}$(S, \Delta)$
          \State $\bar{r}_0 \gets 0$ 
          \For{$i\gets 1,\ldots,k$}
            \State $\bar{r}_i \gets \bar{r}_{i-1} + |S_i| + 1$ \Comment{$\bar{r}_i$ is rank of $P_i$}
            \State $R_i \gets \{r \mid r \in R \;\wedge\; \bar{r}_{i-1} < r < \bar{r}_i\}$ \Comment{Rank queries to bucket $S_i$}
            \If{$|R_i| > 0$}
                \State $r_i^{\max} \gets \max(R_i)$
                \State $\bar{S}_i, p_{\max}, S_{\geq}  \gets \textsc{Select}(S_i, r_i^{\max} - \bar{r}_{i-1})$
                \If{$|R_i| > 1$}
                  \State $r_i^{\min} \gets \min(R_i)$
                    \State $S_{\leq}, p_{\min}, \bar{S}_i  \gets \textsc{Select}(\bar{S}_i, r_i^{\min} - \bar{r}_{i-1})$
                    \State Report $p_   {\min}$
                    \If{$|R_i| > 2$}
                        \State $\bar{R}_i \gets \{r - r_i^{\min} \mid r \in R_i \setminus \{r_i^{\min}, r_i^{\max}\}\}$
                        \State \textsc{DeterministicFunnelselect}$(\bar{S}_i, \bar{R}_i)$
                    \EndIf
                \EndIf
                \State Report $p_{\max}$
            \EndIf
            \If{$r_i \in R$}
              \State Report $P_i$
            \EndIf
          \EndFor
        \EndIf    
      \EndIf
    \EndProcedure
  \end{algorithmic}
  \caption{Deterministic cache-oblivious multiple-selection.}
  \label{alg:deterministic-multiple-selection}
\end{algorithm}
\begin{algorithm}[htp]
  \begin{algorithmic}[1]
    \Procedure{MultiPartition}{$S[1..N]$, $\Delta$}
      \State \textit{Requires $(2N)^{\frac{d}{d+1}} \leq \Delta \leq N$}
      \State $k \gets 1$, $S_1 \gets \{\}$ \Comment{Initially only one empty bucket and no pivots}
      \For{$i\gets1$ \textbf{to} $N$ \textbf{step} $\Delta$}
          \State $\bar{S} \gets S[i..\min(i + \Delta - 1, N)]$ \Comment{Next batch to distribute to buckets}
          \State Distribute $\bar{S}$ to buckets $S_1, \ldots S_k$ using pivots $P_1,\ldots,P_{k-1}$ with a $k$-partitioner
          \While{there exists a bucket $S_j$ with $|S_j| > \Delta$} \Comment{Split bucket $S_j$}
                \State $S_{\leq}, p, S_{\geq} \gets \textsc{Select}(S_j, \lceil |S_j| / 2 \rceil)$
                \State Rename $S_{j+1},\ldots,S_k$ to $S_{j+2},\ldots,S_{k+1}$ and $P_j,\ldots,P_{k-1}$ to $P_{j+1},\ldots,P_k$
                \State $S_j \gets S_{\leq}$, $P_j \gets p$, $S_{j+1} \gets S_{\geq}$
                \State $k \gets k + 1$
          \EndWhile
      \EndFor
      \State \Return $(P_1,\ldots,P_{k-1}), (S_1,\ldots,S_k)$
    \EndProcedure
  \end{algorithmic}
  \caption{Given an array $S$ with $N$ elements and a bucket capacity~$\Delta$, where $(2N)^{\frac{d}{d+1}} \leq \Delta \leq N$, partition~$S$ into $k$ buckets $S_1,\ldots,S_k$ separated by $k-1$ pivots $P_1,\ldots,P_{k-1}$, where $\bigl\lfloor \frac{\Delta}{2}\bigr\rfloor \leq |S_i| \leq \Delta$.}
  \label{alg:multi-partition}
\end{algorithm}

For the case when $\Delta$ is small compared to $N$ (formally, $(2N)^d \geq \Delta^{d + 1}$ or $N^{1+\frac{1}{1+\epsilon}} \geq \Delta^2$), we simply solve the multiple-selection problem by sorting the elements (cache-obliviously using \algname{funnelsort}~\cite{FrigoLPR12}), and report the elements with ranks~$r_1,\ldots,r_q$ by a single scan over the sorted elements. The condition on $\Delta$ implies $\Bio = \Omega(\mathrm{Sort}_{M,B}(N))$, 
where $\mathrm{Sort}_{M,B}(N)=\Theta\bigl(\frac{N}{B}\bigl(1+\log_{M/B} \frac{N}{B}\bigr)\bigr)$ is the number of I/Os required to sort $N$ elements in external memory~\cite{AggarwalVitter88}, 
so this is within a constant factor of the I/O lower bound (detailed analysis in Section~\ref{sec:analysis}).

Otherwise, we create a $k$-partition, where $k=\Theta\bigl(\frac{N}{\Delta}\bigr)$ as follows (\textsc{Multi\-Partition} in \wref{alg:multi-partition}): We repeatedly distribute batches of $\Delta$ elements into a set of buckets separated by pivot elements. Initially we have one empty bucket and no pivot. 
Whenever a bucket reaches size $>\Delta$, the bucket is split into two buckets of size $\leq \Delta$ separated by a new pivot using the (cache-oblivious) linear-time median selection algorithm (\wref{lem:selection}). 
To distribute a batch of elements into the current set of buckets we use a cache-oblivious $k$-partitioner (\wref{lem:k-partitioner}, which depends on the tall-cache assumption parameter~$d$) built using the current set of pivots. 
Note that we need to construct a new $k$-partitioner after each batch of $\Delta$ elements has been distributed, since the number of buckets and pivots can increase. 
For the computation to be I/O efficient, we allocate in memory space for a $\bigl\lfloor\frac{2N}{\Delta}\bigr\rfloor$-partitioner 
followed by space for $\bigl\lfloor\frac{2N}{\Delta}\bigr\rfloor$ buckets of capacity $2\Delta$
(in the proof of \wref{lem:multi-partition} we argue that the number of buckets created is at most $\frac{2N}{\Delta}$ and each bucket will never exceed $2\Delta$ elements). The space for the partitioner is reused for each new batch, and whenever a bucket is split into two new buckets, one bucket remains in the old bucket's allocated space and the other bucket is placed in next available slot for a bucket. This ensures all buckets are stored consecutively in memory, albeit in arbitrary order. 

After having constructed the buckets we compute the ranks of the pivots from the bucket sizes, and consider the gaps with at least one unresolved rank query.  
If the rank of a pivot coincides with a query rank, we report this pivot just after having considered the preceding bucket. Before recursing on the elements in a bucket, we first find the minimum and maximum query ranks $r^{\min}$ and $r^{\max}$ in the bucket by a scan over the bucket's query ranks, and find and report the corresponding elements in the bucket using linear-time selection (\wref{lem:selection}). Finally, we only recurse on the elements between ranks $r^{\min}$ and $r^{\max}$, provided there are any unresolved rank queries to the bucket. 
This ensures that when recursing on a subproblem of size $\bar{N}$, all elements in the subproblem are in gaps of size~$<\bar{N}$ in the original input.
By reporting the elements at the appropriate times during the recursion, elements will be reported in increasing order. 

The partitioning of an array~$S$ into buckets is illustrated in \wref{fig:deterministic-multiple-selection}. The crucial property is that for a gap $\Delta_i\geq \Delta$, the two query ranks~$r_{i-1}$ and $r_{i}$ defining the gap \emph{cannot} be in the same bucket, implying that no element in this gap will be part of a recursive subproblem (see, \eg, gaps $\Delta_1$ and~$\Delta_5$ in \wref{fig:deterministic-multiple-selection}).

Pseudocode for our algorithm is shown in \wref{alg:deterministic-multiple-selection} and \wref{alg:multi-partition}. 
We assume \textsc{Select}$(S, k)$ is the deterministic linear-time selection algorithm from \wref{lem:selection}, and that it returns a triple $S_{\leq}, p, S_{\geq}$, that is a partition of $S$, where $p$ is the element of rank $k$, $S_{\leq}$ are the elements of rank $1,\ldots,k-1$ in arbitrary order, and $S_{\geq}$ the elements of rank $k+1,\ldots,|S|$ in arbitrary order.


\section{Analysis}
\label{sec:analysis}

We first analyze the number of comparisons and I/Os performed by \textsc{MultiPartition} in \wref{alg:multi-partition},
that deterministically performs a $k$-way partition of $N$ elements into $k=O\bigl(\frac{N}{\Delta}\bigr)$ buckets separated by $k-1$ pivots, where each bucket has size at most $\Delta$. 
The following lemma summarizes the precise properties of \textsc{MultiPartition}.

\begin{lemma}
    \label{lem:multi-partition}
    For $N \geq \Delta$ and $\Delta^{d+1} \geq (2N)^d$, \textup{\textsc{MultiPartition}} creates $k\leq \frac{2N}{\Delta}$ buckets and $k-1$ pivots, each bucket has size at most $\Delta$, and performs $O(N\lg k)$ comparisons and $O\bigl(k^2 + \frac{N}{B}(1 + \log_{M} k)\bigr)$ I/Os.
\end{lemma}

\begin{proof}
    We first bound the sizes of the buckets created by \textsc{MultiPartition}. The algorithm repeatedly distributes batches of at most $\Delta$ elements to buckets and splits all overflowing buckets of size $>\Delta$ before considering the next batch. 
    It is an invariant that before distributing a batch, all buckets have size at most~$\Delta$. 
    Furthermore, as soon as the first bucket is split, all buckets have size at least $\bigl\lfloor \frac{\Delta}{2} \bigr\rfloor$, since whenever an overflowing bucket of size $s>\Delta$ is split the new buckets have initial sizes $\bigl\lfloor \frac{s-1}{2} \bigr\rfloor$ and $\bigl\lceil \frac{s-1}{2} \bigr\rceil$. Here ``$-1$'' is due to one element becomes a pivot. 
    The smallest bucket size is achieved when $s=\Delta+1$, where the smallest bucket size is $\bigl\lfloor \frac{\Delta+1-1}{2} \bigr\rfloor= \bigl\lfloor \frac{\Delta}{2} \bigr\rfloor$. 
    Note that the buckets after the split have size at most~$\Delta$, since all buckets had at most $\Delta$ elements before the distribution of a batch of at most $\Delta$ elements to the buckets, \ie, $s \leq 2\Delta$.
    To bound the total number of buckets~$k$ created, observe that if $\Delta = N$ then no bucket will be split and
    $k=1$. Otherwise, $\Delta < N$ and at least two buckets are created, and   
    $k \bigl\lfloor \frac{\Delta}{2} \bigr\rfloor + k-1 \leq N$, since all buckets have size at least $\bigl\lfloor \frac{\Delta}{2} \bigr\rfloor$ and there are $k-1$ pivots.
    We have $N \geq k \bigl(\frac{\Delta}{2} - \frac{1}{2}\bigr) + k - 1 = \frac{k\Delta}{2} + \frac{k}{2} - 1 \geq \frac{k\Delta}{2}$, since $k \geq 2$, \ie, the total number of buckets created $k \leq \frac{2N}{\Delta}$.

    To analyze the number of comparisons and I/Os performed, we need to consider the $\bigl\lceil \frac{N}{\Delta} \bigr\rceil$ distribution steps and at most $\frac{2N}{\Delta} - 1$ bucket splittings. Since each bucket splitting involves at most $2\Delta$ elements, each bucket splitting can be performed cache-obliviously by a linear-time selection algorithm (\wref{lem:selection}) using $O(\Delta)$ comparisons and $O\bigl(1 + \frac{\Delta}{B}\bigr)$ I/Os, assuming each bucket is stored in a buffer of $2\Delta$ consecutive memory cells. In total the $k - 1 = \Theta\bigl(\frac{N}{\Delta}\bigr)$ bucket splittings require $O(N)$ comparisons and $O\bigl(k + \frac{N}{B}\bigr)$ I/Os. 
    A $k$-partitioner for partitioning $\Delta$ elements uses $O(\Delta\lg k)$ comparisons and $O\bigl(k+\frac{\Delta}{B} (1 + \log_M k)\bigr)$ I/Os (\wref{lem:k-partitioner}), assuming $k$ is sufficiently small according to the tall-cache assumption (see below).  This includes the cost of constructing the $k$-partitioner.
    The total cost for all $\bigl\lceil\frac{N}{\Delta}\bigr\rceil$ distribution steps becomes
    $O(N\lg k)$ comparisons and $O\bigl(k\frac{N}{\Delta} + \frac{N}{B}(1+\log_M k)\bigr)=O\bigl(k^2 + \frac{N}{B}(1+\log_M k)\bigr)$ I/Os.

    By \wref{lem:k-partitioner}, the tall-cache assumption $M \geq B^{1+\epsilon}$ implies that for a $k$-partitioner and an input of size $\Delta$, it is required that $\Delta \geq k^d$ for the I/O bounds to hold (recall $d = \max\{1+2/\epsilon,2\}$). The input assumption $\Delta \geq \bigl(\frac{2N}{\Delta} \bigr)^d$ together with $k\leq \frac{2N}{\Delta}$ ensure that $\Delta \geq k^d$.
\end{proof}

We now prove our main result that \textsc{DeterministicFunnelselect} in \wref{alg:deterministic-multiple-selection} is an optimal deterministic cache-oblivious multiple-selection algorithm. Crucial to the analysis is to show that the choice of $\Delta$ balances early pruning of buckets without queries with simultaneously achieving efficient I/O bounds.

\begin{theorem}
    \label{thm:deterministic-funnelselect}
    \textup{\textsc{DeterministicFunnelselect}} performs $O(\Bentropy + N)$ comparisons and $O\bigl(\Bio + \frac{N}{B}\bigr)$ I/Os cache-obliviously in a cache model with tall assumption $M\geq B^{1+\epsilon}$, for some constant~$\epsilon>0$.
\end{theorem}
\begin{proof}
    We first consider the consequences of the choice of $\Delta$. By the choice of $\Delta$, we have $\sum_{\Delta_i < \Delta} \Delta_i < N/2 + 1$.
    Since each bucket $S_i$ has size at most $\Delta$, and we only recurse on subsets that are (the union of) gaps where the two bounding rank queries of the gaps are \emph{both} in the same bucket, we only recurse on gaps with $\Delta_i < \Delta$ elements (see \wref{fig:deterministic-multiple-selection}). 
    A recursive subproblem between query ranks $r_s$ and $r_t$, where $1 \leq s < t \leq q$, contains $r_t-r_s-1 = (\sum_{i=s+1}^t \Delta_i) - 1$ elements. 
    It follows that 
    
    \begin{enumerate}[(A)]
    \item all recursive subproblems in total contain at most $\sum_{\Delta_i < \Delta} \Delta_i - 1 < N/2$ elements and each subproblem has size $\leq \Delta - 2$.
    \item 
    $\sum_{\Delta_i \leq \Delta} \Delta_i \geq N/2 + 1$, \ie, at least $N/2$ elements are in gaps of size at most $\Delta$.
    \end{enumerate}

    To analyze the number of comparisons performed, we use a potential argument where one unit of potential can pay for $O(1)$ comparisons, and all comparisons performed can be charged to the released potential.
    We define the potential of an element $x$ in a gap of size $\Delta_i$ to be $ 1 + \lg \frac{N}{\Delta_i}$, where $N$ is the size of the current recursive subproblem $x$ resides in. The total initial potential is at most $N+\sum_{i=1}^{q+1} \Delta_i \lg \frac{N}{\Delta_i}=O(\Bentropy + N)$. 
    
    We first consider the number of comparisons for the non-sorting case (\wref{alg:deterministic-multiple-selection}, lines~9--26).
    If an element~$x$ in a gap of size $\Delta_i \leq \Delta$ participates in a recursive call of size~$<\Delta$, the potential released for~$x$ is at least $\bigl(1+\lg \frac{N}{\Delta_i}\bigr) - \bigl(1 + \lg \frac{\Delta}{\Delta_i}\bigr) = \lg \frac{N}{\Delta}$. If an element~$x$ in a  gap of size $\Delta_i \leq \Delta$ does not participate in a recursive call, the potential released for~$x$ is $1+\lg \frac{N}{\Delta_i}\geq 1 +\lg \frac{N}{\Delta}$. 
    Finally, elements in gaps of size $> \Delta$ will not participate in recursive calls, and will each release at least potential~1. 
    It follows that the released potential is at least $\frac{N}{2} + \frac{N}{2} \lg \frac{N}{\Delta}$, since at least $N/2$ elements are in gaps of size $\leq \Delta$ (property (B), contributing the second summand) and at most $N/2$ elements are in gaps of size $<\Delta$ and participate in recursive calls (property (A)), \ie, at least $N/2$ elements are in gaps of size $\ge \Delta$ (contributing the first summand). 
    By \wref{lem:multi-partition}, \textsc{MultiPartition} requires $O(N\lg k)$ comparisons, and since $k=O(N/\Delta)$ this can be covered by the released potential. The additional comparisons required for computing~$\Delta$ with a linear-time weighted section algorithm (\wref{lem:weighted-selection}) and performing \textsc{Select} (\wref{lem:selection}) at most twice on each bucket require in total at most $O(N)$ comparisons, and can also be charged to the released potential. It follows that for the non-sorting case the released potential can cover for all comparisons performed.
    
    In the sorting case, a single call to \textsc{Funnelsort} is performed causing $O(N\lg N)$ comparisons (\wref{lem:funnelsort}). No further recursive calls are made and the potential of all elements is released. At least $N + \frac{N}{2} \lg \frac{N}{\Delta}$ potential is released, since at least $N/2$ elements are in gaps of size~$\leq\Delta$ (property (B)). 
    In the sorting case, either $(2N)^d \geq \Delta^{d+1}$ or $N^{1+\frac{1}{1+\epsilon}} \geq \Delta^2$.
    If $(2N)^d \geq \Delta^{d+1}$, we have $\Delta \leq (2N)^{\frac{d}{d+1}}$ and $\frac{N}{\Delta} \geq N/(2N)^{\frac{d}{d+1}} \geq \frac{1}{2} N^{\frac{1}{d+1}}$. It follows that the released potential is at least $N + \frac{N}{2} \lg \bigl(\frac{1}{2} N^{\frac{1}{d+1}}\bigr) \geq \frac{1}{2(d+1)} N\lg N$, covering the cost for the comparisons. 
    Otherwise,  $N^{1+\frac{1}{1+\epsilon}} \geq \Delta^2$, \ie, 
    $\Delta \leq N ^ {\frac{1}{2}\bigl(1+\frac{1}{1+\epsilon}\bigr)}$ and we have
    $\frac{N}{\Delta} 
    \geq N / N ^ {\frac{1}{2}\bigl(1+\frac{1}{1+\epsilon}\bigr)}
    = N^\frac{\epsilon}{2(1+\epsilon)}$
    and the potential released is at least $N + \frac{N}{2} \lg \frac{N}{\Delta} \geq N + \frac{\epsilon}{4(1+\epsilon)}N\lg N$ and can cover the cost for the comparisons.
    Note that the comparison bound depends on the tall-cache parameters~$\epsilon$ and~$d$.
  
    To analyze the I/O cost we assign an I/O potential to an element $x$ in gap of size~$\Delta_i$ of $\frac{1}{B}\bigl(1+\log_M \frac{N}{\Delta_i}\bigr)$, where $N$ is the size of the current subproblem $x$ resides in. 
    Similar to the comparison potential, it follows that the non-sorting case releases I/O potential $\frac{1}{2} \bigl(\frac{N}{B} + \frac{N}{B} \log_M \frac{N}{\Delta}\bigr)$.
    The number of I/Os required is $O\bigl(1 + \frac{N}{B}\bigr)$ I/Os for scanning the input and computing $\Delta$ using weighted selection (\wref{lem:weighted-selection}), $O\bigl(k+\frac{N}{B}\bigr)$ I/Os for selecting the minimum and maximum rank elements in each bucket (\wref{lem:selection}), and $O\bigl(k^2 + \frac{N}{B}(1+\log_M k)\bigr)$ I/Os for the $k$-partitioning (\wref{lem:k-partitioner}), \ie, in total $O\bigl(k^2 + \frac{N}{B}(1 + \log_M k)\bigr)$ I/Os. It follows that the I/O cost can be charged to the released potential, provided $k^2=O\bigl(\frac{N}{B}\bigr)$. 
    To address this, we need to consider two cases depending on the size~$N$ of a subproblem.  If the problem completely fits in internal memory together with all the geometric decreasing recursive subproblems, assuming a stack-oriented memory allocation, then considering this problem will in total cost $O\bigl(1+\frac{N}{B}\bigr)$ I/Os, including all recursive subproblems. That means, there exists a constant $c>0$ such that for $N \leq c M$, the I/O cost for handling such problems can be charged to the parent subproblem creating the subproblem. It follows that we only need to consider the I/O cost for subproblems of size $N \geq cM$. Since $M \geq B^{1+\epsilon}$, we have $N\geq cM\geq c B^{1+\epsilon}$, \ie, 
    $B\leq \bigl(\frac{N}{c}\bigr)^{1/(1+\epsilon)}$. 
    Since $k=O\bigl(\frac{N}{\Delta}\bigr)$, to prove 
    $k^2=O\bigl(\frac{N}{B}\bigr)$ it is sufficient to prove 
    $\bigl(\frac{N}{\Delta}\bigr)^2=O\Bigl(\frac{N}{(N/c)^{1/(1+\epsilon)}}\Bigr)$. 
    This holds, \eg, when $N^{1+\frac{1}{1+\epsilon}} \leq \Delta^2$, which is always fulfilled in the non-sorting case.
    For the sorting case, we have similarly to the comparison potential that $\Omega\bigl(\frac{N}{B}\log_M N\bigr)$ I/O potential is released, which can cover the I/O cost for cache-oblivious sorting (\wref{lem:funnelsort}).
\end{proof}

\section{Conclusion}

With \algname{deterministic funnelselect}, we close the gap left in previous work and obtain an I/O-optimal cache-oblivious multiple-selection algorithm that does not need to resort to randomization to achieve its performance.
This settles the complexity of the multiple-selection problem in the cache-oblivious model (including the fine-grained analysis based on the query-rank entropy $\Bentropy$).

There are open questions left in other variants of the problem.
Like randomized \algname{funnelselect}~\cite{BrodalWild23}, \algname{deterministic funnelselect} cannot deal with queries arriving in an online fashion, one after the other.  This problem has been addressed in the external-memory model~\cite{BarbayGuptaRaoSorenson2016}, but no cache-oblivious I/O-optimal solution is known.

Concerning the transition from single selection by rank to sorting, which multiple selection allows us to study, some questions remain unanswered. 
For example, in the cache-oblivious model, it is known that sorting with optimal I/O-complexity is only possible under a tall-cache assumption (such as the one made in this work); for single selection, however, such a restriction is not necessary. It would be interesting to study the transition between the problems and find out, how ``sorting-like'' a multiple-selection instance has to be to likewise require a tall cache for I/O-optimal cache-oblivious algorithms.

Another direction for future work are parallel algorithms for multiple selection that are also cache-oblivious and I/O efficient.


\bibliographystyle{plainurl}
\bibliography{main}

\begin{thebibliography}{10}

\bibitem{AggarwalVitter88}
Alok Aggarwal and Jeffrey~Scott Vitter.
\newblock The input/output complexity of sorting and related problems.
\newblock {\em Commun. {ACM}}, 31(9):1116--1127, 1988.
\newblock \href {https://doi.org/10.1145/48529.48535}
  {\path{doi:10.1145/48529.48535}}.

\bibitem{BarbayGuptaRaoSorenson2016}
J{\'{e}}r{\'{e}}my Barbay, Ankur Gupta, Srinivasa Rao~Satti, and Jon Sorenson.
\newblock Near-optimal online multiselection in internal and external memory.
\newblock {\em Journal of Discrete Algorithms}, 36:3--17, jan 2016.
\newblock \href {https://doi.org/10.1016/j.jda.2015.11.001}
  {\path{doi:10.1016/j.jda.2015.11.001}}.

\bibitem{BleichOverton83}
Chaya Bleich and Michael~L. Overton.
\newblock A linear-time algorithm for the weighted median problem.
\newblock Technical Report~75, New Yourk University, Department of Computer
  Science, April 1983.
\newblock URL: \url{https://archive.org/details/lineartimealgori00blei/}.

\bibitem{BFPRT73}
Manuel Blum, Robert~W. Floyd, Vaughan~R. Pratt, Ronald~L. Rivest, and
  Robert~Endre Tarjan.
\newblock Time bounds for selection.
\newblock {\em J. Comput. Syst. Sci.}, 7(4):448--461, 1973.
\newblock \href {https://doi.org/10.1016/S0022-0000(73)80033-9}
  {\path{doi:10.1016/S0022-0000(73)80033-9}}.

\bibitem{BF02}
Gerth~St{\o}lting Brodal and Rolf Fagerberg.
\newblock Cache oblivious distribution sweeping.
\newblock In Peter Widmayer, Francisco~Triguero Ruiz, Rafael~Morales Bueno,
  Matthew Hennessy, Stephan~J. Eidenbenz, and Ricardo Conejo, editors, {\em
  Automata, Languages and Programming, 29th International Colloquium, {ICALP}
  2002, Malaga, Spain, July 8-13, 2002, Proceedings}, volume 2380 of {\em
  Lecture Notes in Computer Science}, pages 426--438. Springer, 2002.
\newblock \href {https://doi.org/10.1007/3-540-45465-9_37}
  {\path{doi:10.1007/3-540-45465-9_37}}.

\bibitem{BrodalWild23}
Gerth~St{\o}lting Brodal and Sebastian Wild.
\newblock Funnelselect: Cache-oblivious multiple selection.
\newblock In Inge~Li G{\o}rtz, Martin Farach{-}Colton, Simon~J. Puglisi, and
  Grzegorz Herman, editors, {\em 31st Annual European Symposium on Algorithms,
  {ESA} 2023, September 4-6, 2023, Amsterdam, The Netherlands}, volume 274 of
  {\em LIPIcs}, pages 25:1--25:17. Schloss Dagstuhl - Leibniz-Zentrum f{\"{u}}r
  Informatik, 2023.
\newblock \href {https://doi.org/10.4230/LIPICS.ESA.2023.25}
  {\path{doi:10.4230/LIPICS.ESA.2023.25}}.

\bibitem{Chambers71}
J.~M. Chambers.
\newblock Partial sorting {[M1]} (algorithm 410).
\newblock {\em Commun. {ACM}}, 14(5):357--358, 1971.
\newblock \href {https://doi.org/10.1145/362588.362602}
  {\path{doi:10.1145/362588.362602}}.

\bibitem{DobkinMunro81}
David~P. Dobkin and J.~Ian Munro.
\newblock Optimal time minimal space selection algorithms.
\newblock {\em J. {ACM}}, 28(3):454--461, 1981.
\newblock \href {https://doi.org/10.1145/322261.322264}
  {\path{doi:10.1145/322261.322264}}.

\bibitem{DorZwick1999}
Dorit Dor and Uri Zwick.
\newblock Selecting the median.
\newblock {\em {SIAM} Journal on Computing}, 28(5):1722--1758, 1999.
\newblock \href {https://doi.org/10.1137/s0097539795288611}
  {\path{doi:10.1137/s0097539795288611}}.

\bibitem{FloydRivest1975}
Robert~W. Floyd and Ronald~L. Rivest.
\newblock Expected time bounds for selection.
\newblock {\em Communications of the ACM}, 18(3):165--172, March 1975.
\newblock \href {https://doi.org/10.1145/360680.360691}
  {\path{doi:10.1145/360680.360691}}.

\bibitem{FrigoLPR99}
Matteo Frigo, Charles~E. Leiserson, Harald Prokop, and Sridhar Ramachandran.
\newblock Cache-oblivious algorithms.
\newblock In {\em 40th Annual Symposium on Foundations of Computer Science,
  {FOCS} '99, 17-18 October, 1999, New York, NY, {USA}}, pages 285--298. {IEEE}
  Computer Society, 1999.
\newblock \href {https://doi.org/10.1109/SFFCS.1999.814600}
  {\path{doi:10.1109/SFFCS.1999.814600}}.

\bibitem{FrigoLPR12}
Matteo Frigo, Charles~E. Leiserson, Harald Prokop, and Sridhar Ramachandran.
\newblock Cache-oblivious algorithms.
\newblock {\em {ACM} Trans. Algorithms}, 8(1):4:1--4:22, 2012.
\newblock \href {https://doi.org/10.1145/2071379.2071383}
  {\path{doi:10.1145/2071379.2071383}}.

\bibitem{H61find}
C.~A.~R. Hoare.
\newblock Algorithm 65: find.
\newblock {\em Commun. {ACM}}, 4(7):321--322, 1961.
\newblock \href {https://doi.org/10.1145/366622.366647}
  {\path{doi:10.1145/366622.366647}}.

\bibitem{HuTaoYangZhou2014}
Xiaocheng Hu, Yufei Tao, Yi~Yang, and Shuigeng Zhou.
\newblock Finding approximate partitions and splitters in external memory.
\newblock In {\em Proceedings of the 26th {ACM} symposium on Parallelism in
  algorithms and architectures}. {ACM}, June 2014.
\newblock \href {https://doi.org/10.1145/2612669.2612691}
  {\path{doi:10.1145/2612669.2612691}}.

\bibitem{KMMS05}
Kanela Kaligosi, Kurt Mehlhorn, J.~Ian Munro, and Peter Sanders.
\newblock Towards optimal multiple selection.
\newblock In Lu{\'{\i}}s Caires, Giuseppe~F. Italiano, Lu{\'{\i}}s Monteiro,
  Catuscia Palamidessi, and Moti Yung, editors, {\em Automata, Languages and
  Programming, 32nd International Colloquium, {ICALP} 2005, Lisbon, Portugal,
  July 11-15, 2005, Proceedings}, volume 3580 of {\em Lecture Notes in Computer
  Science}, pages 103--114. Springer, 2005.
\newblock \href {https://doi.org/10.1007/11523468_9}
  {\path{doi:10.1007/11523468_9}}.

\bibitem{Prodinger1995}
Helmut Prodinger.
\newblock Multiple {Q}uickselect -- {H}oare's {F}ind algorithm for several
  elements.
\newblock {\em Information Processing Letters}, 56(3):123--129, November 1995.
\newblock \href {https://doi.org/10.1016/0020-0190(95)00150-b}
  {\path{doi:10.1016/0020-0190(95)00150-b}}.

\bibitem{SchonhagePP76}
Arnold Sch{\"{o}}nhage, Mike Paterson, and Nicholas Pippenger.
\newblock Finding the median.
\newblock {\em J. Comput. Syst. Sci.}, 13(2):184--199, 1976.
\newblock \href {https://doi.org/10.1016/S0022-0000(76)80029-3}
  {\path{doi:10.1016/S0022-0000(76)80029-3}}.

\bibitem{Shamos76}
Michael~Ian Shamos.
\newblock Geometry and statistics: Problems at the interface.
\newblock In Joseph~Frederick Traub, editor, {\em Algorithms and Complexity:
  New Directions and Recent Results}, pages 251--280. Academic Press, 1976.
\newblock URL:
  \url{http://euro.ecom.cmu.edu/people/faculty/mshamos/1976Stat.pdf}.

\end{thebibliography}

\end{document}